\documentclass{article}
\usepackage[colorlinks = true]{hyperref}
\usepackage{amssymb, amsmath, amsthm, tikz, subcaption}

\DeclareMathOperator{\trace}{tr}

\newtheorem{theorem}{Theorem}
\newtheorem{definition}{Definition}

\title{A System of Billiard and Its Application to Information-Theoretic Entropy}
\author{Supriyo Dutta$^1$ $^2$ \thanks{Email: \texttt{dosupriyo@gmail.com}}, Partha Guha $^1$ $^3$ \thanks{Email: \texttt{partha@bose.res.in}} \\ $^1$ S. N. Bose National Centre for Basic Sciences \\ Block - JD, Sector - III, Salt Lake City, Kolkata - 700 106 \vspace{.25 cm} \\ $^2$ Centre for Theoretical Studics \\ Indian Institute of Technology Kharagpur \\ Kharagpur, India - 721302 \vspace{.25 cm}  \\ $^3$ Department of Mathematics \thanks{Address after 1st August 2020} \\ Khalifa University, Abu Dhabi, UAE}
\date{} 

\begin{document}
	\maketitle

	\begin{abstract}
		In this article, we define an information-theoretic entropy based on the Ihara zeta function of a graph which is called the Ihara entropy. A dynamical system consists of a billiard ball and a set of reflectors correspond to a combinatorial graph. The reflectors are represented by the vertices of the graph. Movement of the billiard ball between two reflectors is represented by the edges. The prime cycles of this graph generate the bi-infinite sequences of the corresponding symbolic dynamical system. The number of different prime cycles of a given length can be expressed in terms of the adjacency matrix of the oriented line graph. It also constructs the formal power series expansion of Ihara zeta function. Therefore, the Ihara entropy has a deep connection with the dynamical system of billiards. As an information-theoretic entropy, it fulfils the generalized Shannon-Khinchin axioms. It is a weakly decomposable entropy whose composition law is given by the Lazard formal group law.
	\end{abstract}

	\section{Introduction}
	
		In information theory, entropy is a measure of information. The information is the uncertainty which is inherent in a probability distribution. Shannon entropy is a well-known measure of information. The idea of entropy is diversely studied in the literature of thermodynamics, information theory, dynamical system, graph theory, and social science. The community of mathematical physics is interested to generalize the concept of entropy due to its emerging applications in economics, astrophysics, and informatics. In recent years, the generalization of entropy is a crucial topic for the investigations in mathematics. 
		
		There are different approaches to generalize entropy, in literature. An entropy is a function $S: \{\mathcal{P}\} \rightarrow \mathbb{R}^+ \cup \{0\}$ over the set of all probability distribution $\{\mathcal{P}\}$ satisfying the Shannon-Khincin axioms or their generalizations. The function is independent of the probability distributions. The literature of generalized entropy is concerned with the foundation and properties of the entropy functions. To define new entropy functions we introduce a number of parameters in the expression of Shannon entropy. These parameters may not have any physical significance. The Tsallis entropy is a generalized entropy with a single parameter $q$. Given a discrete probability distribution $\mathcal{P} = \{p_i\}_{i = 1}^W$ the Tsallis entropy is defined by $S_q(\mathcal{P}) = \frac{1}{q - 1} \left(1 - \sum_{i = 1}^W p_i^q\right)$. Observe that $\lim_{q \rightarrow 1} S_q(\mathcal{P}) = S(\mathcal{P}) = - \sum_{i = 1}^W p_i \log(p_i)$, which is the Shannon entropy. Entropy with more than two parameters is also investigated in the literature. Note that, there is no dependence between the parameter $q$ and probability distribution. Therefore, different values of $q$ generate different measures of information for a particular $\mathcal{P}$. Another formulation for generalizing the Shannon entropy is replacing the logarithm with varieties of generalized logarithms, such as deformed logarithms, formal group logarithms, poly-logarithms etc. In this scenario also, the literature is relevant to the properties and the structure of the entropy function.
		
		Following similar ideas, we introduce the Ihara entropy in this article. The Ihara zeta function \cite{terras2010zeta, ihara1966discrete} of a combinatorial graph $G$ is defined by
		\begin{equation}\label{Ihara_zeta_1}
			\zeta_G(z) = \prod_{P}\left(1 - z^{\gamma(P)}\right)^{-1},
		\end{equation}
		where $P$ is a prime cycle in the graph $G$ of length $\gamma(P)$. The Ihara zeta function is defined on a class of graphs satisfying a number of particular characteristics. In this article, we present a physical meaning of these characteristics. Consider the vertices of the graph as reflectors and edges as the movement of a billiard ball between them. It helps us to present the dynamical system as a symbolic dynamical system. The Ihara zeta function acts as a Ruelle zeta function for this system. There are invertible formal power series \cite{dieudonne1973introduction} which can be expressed in terms of Ihara zeta function. We consider one of them as a formal group logarithm, which replaces the natural logarithm for Shannon entropy. The new generalized entropy of probability distributions is mentioned as Ihara entropy, which depends on the structure of graphs. We then prove that this entropy fulfills the Shannon-Khinchin axioms. A number of formal group-theoretic entropy are recently introduced in literature \cite{tempesta2015theorem, tempesta2016beyond, dutta2019ihara}. This article discusses the dynamical system theoretic nature of this entropy. The entropy function depends on the prime cycles of a graph, which are induced by the movement of billiard ball between reflectors. Therefore, the billiard dynamics inherent in the entropy function. Another important characteristic is that the new entropy is a member of the one-parameter class of entropy. This parameter scales a probability distribution in the domain of Ihara zeta function. In addition, this entropy is a measure of uncertainty in a probability distribution and different from the graph entropy or the dynamical entropy. 
		
		This article is distributed into four sections. In section 2, we present a model of billiard dynamics. This section describes a combinatorial graph associated with a billiard dynamical system. It also introduces a symbolic dynamical system where the symbols are the edges of the graph. The bi-infinite sequences of symbols represent the bi-infinite walks, which can be decomposed into prime cycles. The next section is dedicated to the Ihara zeta function and its formal power series representations. Here we define the Ihara entropy and discuss its characteristics. Then we conclude this article.

	\section{A model of billiard dynamics}

		This article considers a particular model of the motion of a billiard ball on a smooth plane. At least four round shaped reflectors are placed at arbitrarily chosen positions on a smooth plain, such that, they are not arranged on a single straight line. A billiard ball moves between the reflectors and reflected elastically when it collides with a reflector. The ball can not be reflected on the same reflector consecutively. We are not interested in the radius of the reflectors, their internal distance, initial and terminal position of the ball, as well as initial speed and angles of reflections of the billiard.
		
		To associate a combinatorial graph with this system we assume the reflectors as the vertices. There is an edge between two vertices if a ball can be reflected between the corresponding reflectors. A ball can move in any directions, between two reflectors. Therefore, each edge has two opposite orientations. It is assumed that the ball can not be consequently reflected on the same reflector. Thus, there is no loop on the vertices. The reflectors do not form a straight line. Hence, the path graphs are excluded from our consideration. This assumption also indicates that there is no vertex in these graphs which is adjacent to only one vertex. A cycle is also excluded from our discussion. We can arrange the reflectors in a cycle. In this arrangement, if a ball moves from the reflectors $a$ to another reflector $b$, then it has a chance to move towards another reflector $c$ which is located nearly $b$. Combining all these observations, we find that a graph $G = (V(G), E(G))$ describing the dynamics of billiard under our consideration, is a simple, finite, connected, and undirected graph without any vertex of degree one. In addition, $G$ is neither a cycle graph nor a path graph. We call them admissible graphs. Different arrangements of the reflectors are represented by different admissible graphs. As an example, consider figure \ref{system_of_reflectors} which contains a set of reflectors which are represented by circles. This system can be represented as the combinatorial graph depicted in figure \ref{graph_representation}.
		
		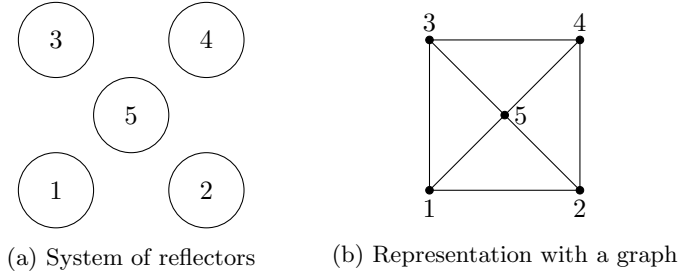
\begin{figure}
			\centering 
			\begin{subfigure}{0.4\textwidth}
				\centering 
				\begin{tikzpicture}
				\draw (0, 0) circle [radius= .5cm];
				\node at (0, 0) {$1$};
				\draw (2, 0) circle [radius= .5cm];
				\node at (2, 0) {$2$};
				\draw (0, 2) circle [radius= .5cm];
				\node at (0, 2) {$3$};
				\draw (2, 2) circle [radius= .5cm];
				\node at (2, 2) {$4$};
				\draw (1, 1) circle [radius= .5cm];
				\node at (1, 1) {$5$};
				\end{tikzpicture}
				\caption{System of reflectors} 
				\label{system_of_reflectors} 
			\end{subfigure}
			\begin{subfigure}{.4\textwidth}
				\centering 
				\begin{tikzpicture}
				\draw [fill] (0, 0) circle [radius= .05];
				\node [below] at (0, 0) {$1$};
				\draw [fill] (2, 0) circle [radius= .05];
				\node [below] at (2, 0) {$2$};
				\draw [fill] (0, 2) circle [radius= .05];
				\node [above] at (0, 2) {$3$};
				\draw [fill] (2, 2) circle [radius= .05];
				\node [above] at (2, 2) {$4$};
				\draw [fill] (1, 1) circle [radius= .05];
				\node [right] at (1, 1) {$5$};
				\draw (0, 0) -- (0, 2) -- (2, 2) -- (2, 0) -- (0, 0) -- (1, 1) -- (2, 2);
				\draw (2, 0) -- (1, 1) -- (0, 2);
				\end{tikzpicture}
				\caption{Representation with a graph} 
				\label{graph_representation}
			\end{subfigure} 
			\caption{Reflectors and the corresponding graph}
		\end{figure}
		
		Let a particular billiard system be represented by a graph $G$ with $n$ vertices and $m$ edges. As every edge has two opposite orientations the set of all orientations can be collected as 
		\begin{equation}\label{alphabet}
		\mathcal{E} = \{e^{(1)}, e^{(2)}, \dots e^{(m)}, e^{(m+1)} = (e^{(1)})^{-1}, \dots e^{(2m)} = (e^{(m)})^{-1} \}.
		\end{equation}
		Here, $e^{(k)} = (u,v)$ is a directed edge with initial and terminating vertices $i(e^{(k)}) = u$ and $t(e^{(k)}) = v$, respectively. Every directed edge $e^{(k)}$ has an inverse $e^{(m + k)} = (e^{(k)})^{-1} = (v, u)$ in $\mathcal{E}$ for $k = 1, 2, \dots m$.
	
		The movement of a ball between the reflectors generates a directed path, which is a sequence of directed edges, in the graph. We are not interested in the initial and terminal position of the billiard ball. Hence, we assume that the sequence of directed edges forms a bi-infinite, directed walk on the graph. Two directed edges $e^{(i)}$ and $e^{(j)}$ consecutively arise on a walk if $v = t(e^{(i)}) = i(e^{(j)})$. It indicates that a ball which is moving along the direction of $e^{(i)}$ will follow the direction of $e^{(j)}$ after getting reflected at $v$. We describe that the edges $e^{(i)}$ and $e^{(j)} \in \mathcal{E}$ are composeble and the composition is represented by $e^{(i)}e^{(j)}$. Note that, the set $\mathcal{E}$ forms a collections of symbols, that is an alphabet \cite{lind1995introduction} of a symbolic dynamic system. Symbolically, we write a bi-infinite walk $w = \prod_{i \in \mathbb{Z}} e_i = \dots e_{-2} e_{-1} e_0 e_1 e_2 \dots$, such that, any two constitutive edges $e_i$ and $e_{i + 1}$ are composeble for all $i \in \mathbb{Z}$. The set of all such walks is a full $\mathcal{E}$-shift, which is denoted by $\mathcal{E}^{\mathbb{Z}}$. A block over $\mathcal{E}$ is a walk $Q = e_1e_2\dots e_{\gamma(Q)}$ of finite length $\gamma(Q)$. The ball rarely reflects between two reflectors repeatedly. Therefore, we neglect situation $e_{i+1} = (e_{i})^{-1}$, in a bi-infinite walk. Now we define a set of forbidden blocks $\mathcal{F} = \{e^{(i)}(e^{(i)})^{-1}: e^{(i)} \in \mathcal{E}\}$. Let $X_{\mathcal{F}}$ be the subset of $\mathcal{E}^\mathbb{Z}$ which does not contain any block in $\mathcal{F}$. Note that, $X_\mathcal{F}$ is a shift space of finite type. A cycle $W = e_1e_2\dots e_k$ of length $k$ is a closed finite walk, such that, $i(e_1) = t(e_k)$. Two cycles $W_{1} = e_{1,1} e_{1,2} \dots e_{1,k}$ and $W_{2} = e_{2,1} e_{2,2} \dots e_{2,k}$ are equivalent if $e_{2,1} = e_{1, r}, e_{2,2} = e_{1, (r+1)}, \dots e_{2,(k - r + 1)} = e_{1,k}, e_{2,(k - r + 2)} = e_{1,1}, \dots e_{2,k} = e_{1,(r - 1)}$ for some $r \in \{1, 2, \dots k\}$. The set of equivalence classes of cycles are called prime cycles. The length of a prime cycle $P$ is denoted by $\gamma(P)$. Two primes $P_1 = e_{1,1} e_{1,2} \dots e_{1, k_1}$ and $P_2 = e_{2,1} e_{2,2} \dots e_{2, k_2}$ are composable if $t(e_{1, k_1}) = i(e_{2, k_2})$, and the composition is denoted by $P^{(1)}P^{(2)} = e_{1,1} e_{1,2} \dots e_{1, k_1}e_{2,1} e_{2,2} \dots e_{2, k_2}$. In a similar fashion, we define product of a prime $P$ and a finite walk $Q$, or product of two finite walks $Q_1$ and $Q_2$. A simple observation indicates that any bi-infinite walk $w \in \mathcal{E}^{\mathbb{Z}}$ can be expressed as $w = \prod_{i \in \mathbb{Z}} P_{i}^{p_i} Q_i$, where $p_i$ is the number of consecutive repetitions of the prime $P_i$ and $Q_i$ is the finite walk after the repetition of prime $P_i$.

	\section{Ihara zeta function and entropy}
		
		Recall from the last section that the movement of the billiard ball in the system of reflectors generates bi-infinite walks on a graph $G$. To illustrate the properties of these walks we consider the oriented edges as the vertices of a new graph, which is called the oriented line graph. Formally, the oriented line graph $\overline{G} = (V(\overline{G}), E(\overline{G}))$ of the graph $G$ is represented by $V(\overline{G}) = \mathcal{E}$ and
		\begin{equation}
			E(\overline{G})  = \{(e^{(i)}, e^{(j)}) \in \mathcal{E} \times \mathcal{E} : t(e^{(i)}) = i(e^{(j)}) ~\text{and}~ i(e^{(i)}) \neq t(e^{(j)})\}. 
		\end{equation}
		It is known that the number of prime cycles of length $k$ starting and ending at the vertex $e$ is expressed as the $e$-th element of the $T^k$, where $T = (t_{(e^{(i)}, e^{(j)})})_{2m \times 2m}$ is the adjacency matrix of the graph $\overline{G}$ defined by,
		\begin{equation}
			t_{(e^{(i)}, e^{(j)})} = \begin{cases} 1 & \text{if}~ (e^{(i)}, e^{(j)}) \in E(\overline{G}), ~\text{and}\\ 0 & \text{if}~ (e^{(i)}, e^{(j)}) \notin E(\overline{G}). \end{cases} 
		\end{equation}
		Therefore, $\trace(T^k)$ represents the number of all cycles of length $k$, which is a non-negative integer. Now the generating function for the number of cycles in a graph is given by $f(z) = \sum_{k = 1}^\infty \frac{\trace(T^k)}{k}z^k$. The Ihara zeta function for the graph $G$ is alternatively represented by the formal power series 
		\begin{equation}\label{Ihara_zeta_2}
		\zeta_G(z) = \exp\left(\sum_{k = 1}^\infty \frac{\trace(T^k)}{k}z^k\right),
		\end{equation}
		where $|z| < \frac{1}{\lambda}$ \cite{kotani20002}. Here, $\lambda$ is the greatest eigenvalue of $T$, which is a positive number.
		
		In this work, we are interested in entropy of a probability distribution depending on the billiard dynamics. As probability is a positive real number, we restrict $\zeta_G(z)$ to the real interval $[0, \frac{1}{\lambda})$. The restricted function $\zeta_G(x) = \exp\left(\sum_{k = 1}^\infty \frac{\trace(T^{k})}{k}x^{k}\right)$, such that, $\zeta_G(x): [0, \frac{1}{\lambda}) \rightarrow \mathbb{R}$ can be expressed as
		\begin{equation}\label{zeta_with_real_coefficients}
			\begin{split} 
				\zeta_G(x) = & 1 + \sum_{k = 1}^\infty \frac{\trace(T^{k})}{k}x^{k} + \frac{1}{2!}\left(\sum_{k = 1}^\infty \frac{\trace(T^{k})}{k}x^{k}\right)^2 + \frac{1}{3!} \left(\sum_{k = 1}^\infty \frac{\trace(T^{k})}{k}x^{k} \right)^3 + \dots \\
				= & 1 + c_1x + c_2x^2 + c_3x^3 + c_4x^4 + c_5x^5 + \dots. 
			\end{split} 
		\end{equation}
		In the above expression, $c_1 = \trace(T) = 0$, since $T$ is an adjacency matrix of a graph without a loop. Also, $c_2 = \frac{\trace(T^2)}{2}, c_3 = \frac{\trace{T^3}}{3}, c_4 =  \frac{\trace(T^4)}{4} + \frac{(\trace(T^2))^2}{8}, c_5 = \frac{\trace(T^5)}{5} + \frac{\trace(T^2)\trace(T^3)}{6}$. As, $\trace{T^k}$ is non-negative for all $k$, and the coefficients $c_2, c_3, \dots$ of $\zeta_G(x)$ in equation (\ref{zeta_with_real_coefficients}) are all positive. Hence, for all $x \in [0, \frac{1}{\lambda})$ $\zeta_G(x) > 0$, as well as all its derivatives exists and positive. Clearly, $\zeta_G(x), \zeta'_G(x), \zeta''_G(x), \dots$ are all monotone increasing functions.
		
		Given two formal power series \cite{brewer2014algebraic} $S = \sum_{i = 1}^\infty s_i x^i$ and $T = \sum_{i = 1}^\infty t_i x^i$ the composition $S\circ T(x)$ is defined by another power series $S\circ T (x) = S(T(x))$. The power series $S$ is said to be the compositional inverse of $T$ if $S(T(x)) = x$ holds. The power series $T$ has an inverse with respect to the composition if and only if $t_1 = 1$. The coefficients in equation (\ref{zeta_with_real_coefficients}) suggest that the formal power series of  $\zeta_G(x)$ has no compositional inverse.
		
		The formal group entropy \cite{tempesta2015theorem} of a discrete probability distribution $\mathcal{P} = \{p_i\}_{i = 1}^W$ is given by $S(\mathcal{P}) = \sum_{i = 1}^W p_i \mathcal{G}\left(\log\left(\frac{1}{p_i}\right)\right)$, where $\mathcal{G}$ is an invertible formal power series. Let $t = \log(\frac{1}{p})$, which refers $p = e^{-t}$. Now, the equation (\ref{zeta_with_real_coefficients}) indicates that
		\begin{equation}
			\zeta_G(ae^{-t}) = 1 + c_2 (ae^{-t})^2 + c_3 (ae^{-t})^3 + c_4(ae^{-t})^4 + \dots.
		\end{equation}
		Here, $a$ is a non-zero scaling factor, such that, $0 \leq ae^{-t} < \frac{1}{\lambda}$. In addition,
		\begin{equation}
			\zeta_G(a) = 1 + c_2 a^2 + c_3 a^3 + c_4 a^4 + \dots.
		\end{equation}
		Hence,
		\begin{equation}
			\zeta_G(ae^{-t}) - \zeta_G(a) = c_2 a^2(e^{-2t} - 1) + c_3 a^3(e^{-3t} - 1) + c_4 a^4(e^{-4t} - 1) + \dots 
		\end{equation}
		Note that, $e^{-t} - 1 = -t + \frac{t^2}{2!} - \frac{t^3}{3!} + \dots$. Now,
		\begin{equation}
			\begin{split} 
				& \zeta_G(ae^{-t}) - \zeta_G(a) + e^{-t} - 1\\
				= & (e^{-t} - 1) + c_2 a^2(e^{-2t} - 1) + c_3 a^3(e^{-3t} - 1) + c_4 a^4(e^{-4t} - 1) + \dots
		\end{split} 
		\end{equation}
		Clearly, $\zeta_G(ae^{-t}) - \zeta_G(a) + e^{-t} - 1$ has no constant term. The coefficient of $t$ in the power series of $\zeta_G(ae^{-t}) - \zeta_G(a) + e^{-t} - 1$ is 
		\begin{equation}
			\begin{split}
				& \frac{d}{dt}\left[\zeta_G(ae^{-t}) - \zeta_G(a) + e^{-t} - 1\right]|_{t = 0}\\
				= & \left[-a e^{-t} \zeta_G'(ae^{-t}) - e^{-t} \right] |_{t = 0} = - \left[1 + a \zeta'_G(a)\right].
			\end{split}
		\end{equation}
		Hence, the formal power series corresponding to 
		\begin{equation}\label{formal_power_series_for_G}
			\mathcal{G}(t) = \frac{\zeta_G(ae^{-t}) - \zeta_G(a) + e^{-t} - 1}{- (1 + a \zeta'_G(a))}
		\end{equation}
		has zero constant coefficient as well as the coefficient for $t$ is $1$. Therefore, there exists a formal power series $\mathcal{F}(s)$, such that, $\mathcal{F}(\mathcal{G}(t)) = \mathcal{G}(\mathcal{F}(t)) = t$. Now, replacing $t = \log(\frac{1}{p})$ in the expression of $\mathcal{G}(t)$ we find 
		\begin{equation}
			\mathcal{G}\left(\log\left(\frac{1}{p}\right)\right) = \frac{\zeta_G(a) + 1 - (\zeta_G(ap) + p) }{1 + a \zeta'_G(a)}.
		\end{equation}
		As $\zeta_G(x)$ is a monotone increasing function, $\zeta_G(a) \geq \zeta_G(ap) > 0$. Therefore, $\mathcal{G}\left(\log\left(\frac{1}{p}\right)\right) \geq 0$. It leads us to construct the formal group theoretic entropy associated to the Ihara zeta function, which is defined below. 
		\begin{definition}\label{One_parameter_trace_from_entropy}
			Given a graph $G$ the Ihara entropy of a discrete probability distribution $\mathcal{P} = \{p_i\}_{i = 1}^W$ is defined by
			$$S_G(\mathcal{P}) = \sum_{i = 1}^W p_i \mathcal{G} \left( \log\left( \frac{1}{p_i} \right)\right) = \sum_{i = 1}^W p_i \frac{\zeta_G(a) + 1 - (\zeta_G(ap_i) + p_i) }{1 + a \zeta'_G(a)},$$
			where $0 < a < \frac{1}{\lambda}$. Here, $\lambda$ is the largest eigenvalue of $\overline{G}$, which is the oriented line graph of $G$.
		\end{definition}
		
		In information theory, an entropy $S(\mathcal{P})$ of a probability distribution $\mathcal{P}$ satisfies the Shannon-Khinchin axioms \cite{shannon1948mathematical, khinchin2013mathematical} which are mentioned below:
		\begin{enumerate}
			\item \label{axiom_1} 
			The function $S([\mathcal{P}])$ is continuous with respect to all its arguments $p_i$, where $[\mathcal{P}] = \{p_i\}_{i = 1}^W$ is a discrete probability distribution.
			\item \label{axiom_2}
			Adding a zero probability event to a probability distribution does not alter its entropy, that is $S([\mathcal{P}_1]) = S([\mathcal{P}])$ where $[\mathcal{P}_1] = \{p_i\}_{i = 1}^W \cup \{0\}$.
			\item \label{axiom_3}
			The function $S([\mathcal{P}])$ is maximum for the uniform distribution $[\mathcal{P}] = \{\frac{1}{W}\}_{i = 1}^W$.
			\item \label{axiom_4}
			Given two independent subsystems $A, B$ of a statistical system, $S(A + B) = S(A) + S(B)$.
		\end{enumerate}
		
		Define a function $s:[0, 1] \rightarrow \mathbb{R}$, such that,
		\begin{equation}\label{s_p}
			s(p) = p \times \frac{\zeta_G(a) + 1 - (\zeta_G(ap) + p) }{1 + a \zeta'_G(a)}.
		\end{equation}
		Therefore, the Ihara entropy $S_G(\mathcal{P}) = \sum_{i = 1}^W s(p)$. Clearly, $s(p)$ is a continuous function of $p$, that is, $S_G(\mathcal{P})$ is also continuous with respect to all its arguments $p_i$ for $i = 1, 2, \dots W$. Thus, $S_G(\mathcal{P})$ satisfies the axiom \ref{axiom_1}. The axiom \ref{axiom_2} also trivially satisfied as $s(0) = 0$, that is $0$ probability alters nothing in $S(\mathcal{P})$. The axiom \ref{axiom_3} and axiom \ref{axiom_4} are non-trivial which are illustrated in the following two theorems. 

		\begin{theorem}\label{unique_global maxiam}
			There exists a global maxima of $s(p)$ in $(0, 1)$, where $s(p)$ is defined in equation (\ref{s_p}).
		\end{theorem}
		\begin{proof}  
			We have
			\begin{equation}
				s'(p) = \frac{1 + \zeta_G(a) - 2p - \zeta_G(a p) - a p \zeta'_G(a p)}{1 + a\zeta_G'(a)}.
			\end{equation}
			Now $s'(p) = 0$ holds if and only if 
			\begin{equation}
				h(p) = 1 + \zeta_G(a) - 2p - \zeta_G(a p) - a p \zeta'_G(a p)
			\end{equation}
			has a root in $(0,1)$. Equation (\ref{zeta_with_real_coefficients}) suggests that for any graph $G$ we have $\zeta_G(0) = 1$. Therefore, $h(0) = \zeta_G(a) > 0$ and $h(1) = -1 - a \zeta_G'(a) < 0$. As $h$ is a continuous function of $p$ there is at least one point $p = c$ in $(0, 1)$, such that, $h(c) = 0$ that is $s'(c) = 0$. Also, $h'(p) = -2 - 2a \zeta'_G(a p) - a^2 p \zeta''_G(a p) < 0$ for all $p$. Therefore, $h$ is strictly monotone decreasing function, that is $c$ is the unique point in $(0, 1)$ such that $s'(c) = 0$. Now, 
			\begin{equation}
				s''(p) = \frac{-2 - 2a \zeta'_G(a p) - a^2 p \zeta''_G(a p)}{1 + a\zeta_G'(a)} < 0,
			\end{equation}
			for all $p$. Hence, $c$ is a global maxima of $s(p)$ in $(0, 1)$.
		\end{proof} 

		The theorem \ref{unique_global maxiam} leads us to the conclusion that the entropy $S_G(\mathcal{P})$ considers the maximum value if $s(p_i)$ is maximum for all $p_i \in \mathcal{P}$. Thus, to maximize $S_G(\mathcal{P})$ we need $p_i = c$ for all $i$, which is the uniform distribution after a normalization. Therefore, the Ihara entropy mentioned in definition \ref{One_parameter_trace_from_entropy} fulfills the axiom \ref{axiom_3} of the Shannon-Khinchin axioms.
		
		We generalize the axiom \ref{axiom_4} of the Shannon-Khinchin axioms by utilize the Lazard formal group law. Recall that, a commutative one-dimensional formal group law over $\mathbb{R}$  is a formal power series $\Phi(x, y)$ with two indeterminates $x$ and $y$ of the form $\Phi(x, y) = x + y +$ higher order terms, such that
		\begin{equation}
			\begin{split}
				& \Phi(x, 0) = \Phi(0, x) = x, \\
				& \Phi (\Phi (x, y), z) = \Phi (x, \Phi (y, z)),\\
				& \Phi(x, y) = \Phi(y, x).
			\end{split}
		\end{equation}
		Recall from equation (\ref{formal_power_series_for_G}) that we have considered the $\mathcal{F}(s)$ as the compositional inverse of $\mathcal{G}(t)$. Now, the Lazard formal group law \cite{hazewinkel1978formal} is defined by the formal power series
		\begin{equation}\label{Lazard_formal_group_law} 
		\Phi(s_1, s_2) = \mathcal{G} \left( \mathcal{F}(s_1 ) + \mathcal{F}(s_2) \right).
		\end{equation}

		\begin{theorem}
			Let $\mathcal{P}_A = \left\{p_i^{(A)} \right\}_{i = 1}^{W_A}$ and $\mathcal{P}_B = \left\{p_j^{(B)}\right\}_{j = 1}^{W_B}$ be two independent probability distributions. Then the Ihara entropy of the joint probability distribution is given by
			$$S_G(\mathcal{P}_A \mathcal{P}_B) = \sum_{i = 1}^{W_A} \sum_{j = 1}^{W_B} p_i^{(A)} p_j^{(B)} \Phi \left(\mathcal{G} \left(\log \left(\frac{1}{p_{i}^{(A)}}\right) \right), \mathcal{G} \left(\log \left(\frac{1}{p_{j}^{(B)}}\right) \right) \right),$$
			where $\Phi$ is Lazard formal group law given by $\Phi(s_1, s_2) = \mathcal{G} \left( \mathcal{F}(s_1 ) + \mathcal{F}(s_2) \right)$.
		\end{theorem}
		\begin{proof}
			The joint probability distribution $\mathcal{P}_A \mathcal{P}_B$ is given by $p_{ij}^{(A\cup B)} = p_i^{(A)} p_j^{(B)}$. Now,
			\begin{equation}
				\begin{split}
					S_G(\mathcal{P}_A \mathcal{P}_B) & = \sum_{i = 1}^{W_A} \sum_{j = 1}^{W_B} p_{ij}^{(A\cup B)} \mathcal{G} \left( \log \left(\frac{1}{p_{ij}^{(A\cup B)}}\right) \right) \\
					& = \sum_{i = 1}^{W_A} \sum_{j = 1}^{W_B} p_i^{(A)} p_j^{(B)} \mathcal{G} \left( \log \left(\frac{1}{p_{i}^{(A)}}\right) + \log \left(\frac{1}{p_{j}^{(B)}}\right) \right).
				\end{split}
			\end{equation}
			Denote $\log \left(\frac{1}{p_{i}^{(A)}}\right) = t_i^{(A)}$ and $\log \left(\frac{1}{p_{j}^{(B)}}\right) = t_j^{(B)}$. It leads us to write
			\begin{equation}
				\begin{split}
					S_G(\mathcal{P}_A \mathcal{P}_B) & = \sum_{i = 1}^{W_A} \sum_{j = 1}^{W_B} p_i^{(A)} p_j^{(B)} \mathcal{G} \left(t_i^{(A)} + t_j^{(B)} \right) \\
					& = \sum_{i = 1}^{W_A} \sum_{j = 1}^{W_B} p_i^{(A)} p_j^{(B)} \mathcal{G} \left(\mathcal{F} \left(s_i^{(A)} \right) + \mathcal{F} \left(s_j^{(B)}\right) \right),
				\end{split}
			\end{equation}
			where $\mathcal{F}$ is the compositional inverse of $\mathcal{G}$ as well as $\mathcal{F} \left(s_i^{(A)} \right) = t_i^{(A)}$ and $\mathcal{F} \left(s_j^{(B)} \right) = t_j^{(B)}$. Applying Lazard formal group law we have 
			\begin{equation}
				\begin{split}
					S_G(\mathcal{P}_A \mathcal{P}_B) & = \sum_{i = 1}^{W_A} \sum_{j = 1}^{W_B} p_i^{(A)} p_j^{(B)} \Phi \left(s_i^{(A)}, s_j^{(B)}\right) \\
					& = \sum_{i = 1}^{W_A} \sum_{j = 1}^{W_B} p_i^{(A)} p_j^{(B)} \Phi \left(\mathcal{F}^{-1} \left(t_i^{(A)} \right), \mathcal{F}^{-1} \left(t_j^{(B)} \right) \right)\\
					& = \sum_{i = 1}^{W_A} \sum_{j = 1}^{W_B} p_i^{(A)} p_j^{(B)} \Phi \left(\mathcal{G} \left(\log \left(\frac{1}{p_{i}^{(A)}}\right) \right), \mathcal{G} \left(\log \left(\frac{1}{p_{j}^{(B)}}\right) \right) \right).
				\end{split}
			\end{equation}
		\end{proof}
	
		The axiom \ref{axiom_4} is generalised by the composition law of the Lazard formal group law mentioned in equation  (\ref{Lazard_formal_group_law}) \cite[theorem 1]{tempesta2016beyond}.

	\section{Conclusion}

		This article is at the interface of the dynamical system, information and graph theory. It focuses on the information-theoretic entropy of a discrete probability distribution.  This article has a two-fold significance. It presents a physical significance for selecting a particular class of graphs in the literature of the Ihara zeta function. We begin with a dynamical system consists of a billiard ball moving between the reflectors. We describe the reflectors as the vertices of a combinatorial graph. An edge between two vertices represents the possibility of movement of the ball between the corresponding reflectors. A bi-infinite path generated by the movement of the ball represents a bi-infinite walk in the graph. Every bi-infinite walk can be decomposed into prime cycles in the graph. The number of prime cycles of finite length can be expressed in terms of the adjacency matrix of an oriented line graph. We can represent this system in terms of symbolic dynamics over the corresponding graph. The Ihara zeta function is the dynamical zeta function for this system. It can be represented as a formal power series. Note that, this formal power series depends on the distribution of reflectors in the system.  Now the idea of Ihara entropy is introduced. It is an entropy in terms of the Ihara zeta function. The composition law of this entropy is induced by the Lazard formal group law. It also satisfy the other properties of the Shannon-Khinchin axioms.

	\section*{Acknowledgment}

		SD is thankful to Dr. Subhashish Banerjee who introduced the author to the Ihara Zeta function and its applications in quantum information theory.


\begin{thebibliography}{10}
		
		\bibitem{terras2010zeta}
		Audrey Terras.
		\newblock {\em Zeta functions of graphs: a stroll through the garden}, volume
		128.
		\newblock Cambridge University Press, 2010.
		
		\bibitem{ihara1966discrete}
		Yasutaka Ihara.
		\newblock On discrete subgroups of the two by two projective linear group over
		p-adic fields.
		\newblock {\em Journal of the Mathematical Society of Japan}, 18(3):219--235,
		1966.
		
		\bibitem{dieudonne1973introduction}
		Jean~A Dieudonne.
		\newblock {\em Introduction to the theory of formal groups}, volume~20.
		\newblock CRC Press, 1973.
		
		\bibitem{tempesta2015theorem}
		Piergiulio Tempesta.
		\newblock A theorem on the existence of trace-form generalized entropies.
		\newblock {\em Proc. R. Soc. A}, 471(2183):20150165, 2015.
		
		\bibitem{tempesta2016beyond}
		Piergiulio Tempesta.
		\newblock Beyond the shannon--khinchin formulation: the composability axiom and
		the universal-group entropy.
		\newblock {\em Annals of Physics}, 365:180--197, 2016.
		
		\bibitem{dutta2019ihara}
		Supriyo Dutta and Partha Guha.
		\newblock Ihara zeta entropy.
		\newblock {\em arXiv preprint arXiv:1906.02514}, 2019.
		
		\bibitem{lind1995introduction}
		Douglas Lind, Brian Marcus, Lind Douglas, and Marcus Brian.
		\newblock {\em An introduction to symbolic dynamics and coding}.
		\newblock Cambridge university press, 1995.
		
		\bibitem{kotani20002}
		Motoko Kotani and Toshikazu Sunada.
		\newblock 2.-zeta functions of finite graphs.
		\newblock {\em Journal of Mathematical Sciences-University of Tokyo},
		7(1):7--26, 2000.
		
		\bibitem{brewer2014algebraic}
		Thomas~S Brewer.
		\newblock Algebraic properties of formal power series composition.
		\newblock 2014.
		
		\bibitem{shannon1948mathematical}
		Claude~Elwood Shannon.
		\newblock A mathematical theory of communication.
		\newblock {\em Bell system technical journal}, 27(3):379--423, 1948.
		
		\bibitem{khinchin2013mathematical}
		A~Ya Khinchin.
		\newblock {\em Mathematical foundations of information theory}.
		\newblock Courier Corporation, 2013.
		
		\bibitem{hazewinkel1978formal}
		Michiel Hazewinkel.
		\newblock {\em Formal groups and applications}, volume~78.
		\newblock Elsevier, 1978.
		
	\end{thebibliography}

\end{document}